\documentclass[11pt]{amsart}

\usepackage{amssymb,latexsym}
\usepackage{graphicx,epsfig,color}

\newtheorem{thm}{Theorem}[section]

\newtheorem{cor}[thm]{Corollary}
\newtheorem{lem}[thm]{Lemma}
\newtheorem{rem}[thm]{Remark}

\newtheorem{defn}[thm]{Definition}

\numberwithin{equation}{section}

\date{}

\begin{document}

\author[Tulkin H. Rasulov, Elyor B. Dilmurodov]
{Tulkin H. Rasulov, Elyor B. Dilmurodov}
\title[Eigenvalues and virtual levels of a family of $2 \times 2$ operator matrices]
{Eigenvalues and virtual levels of a family of $2 \times 2$ operator matrices}

\maketitle
\begin{center}
{\small Department of Mathematics\\
Faculty of Physics and Mathematics\\
Bukhara State University\\
M. Ikbol str. 11, 200100 Bukhara, Uzbekistan\\
E-mail: rth@mail.ru, elyor.dilmurodov@mail.ru}
\end{center}

\begin{abstract}
In the present paper we consider a family of $2 \times 2$ operator matrices ${\mathcal A}_\mu(k),$
$k \in {\Bbb T}^3:=(-\pi, \pi]^3,$ $\mu>0,$ associated with the Hamiltonian of a system consisting
of at most two particles on a three-dimensional lattice ${\Bbb Z}^3,$ interacting via creation and annihilation
operators. We prove that there is a value $\mu_0$ of the parameter $\mu$ such that only for $\mu=\mu_0$ the
operator ${\mathcal A}_\mu(\overline{0})$ has a virtual level at the
point $z=0=\min\sigma_{\rm ess}({\mathcal A}_\mu(\overline{0}))$ and the
operator ${\mathcal A}_\mu(\overline{\pi})$ has a virtual level at the
point $z=18=\max\sigma_{\rm ess}({\mathcal A}_\mu(\overline{\pi}))$,
where $\overline{0}:=(0,0,0), \overline{\pi}:=(\pi,\pi,\pi) \in {\Bbb T}^3.$
The absence of the eigenvalues of ${\mathcal A}_\mu(k)$ for all values of $k$ under the assumption
that $\mu=\mu_0$ is shown. The threshold energy expansions for the Fredholm determinant associated
to ${\mathcal A}_\mu(k)$ are obtained.
\end{abstract}



\section{\bf Introduction}

Block operator matrices are matrices where the entries are
linear operators between Banach or Hilbert spaces \cite{CT08}. One
special class of block operator matrices are Hamiltonians
associated with systems of non-conserved number of quasi-particles
on a lattice. Their number can be unbounded as in the case of spin-boson models \cite{HSp95, HS89}
or bounded as in the case of "truncated" $ $ spin-boson models \cite{MinSp96, MNR15, TR16}.
They arise, for example, in the theory of solid-state physics \cite{Mog91},
quantum field theory \cite{Frid65} and statistical physics
\cite{MalMin95, MinSp96}.

The eigenvalues and virtual levels of block operator matrices, in particular,
Hamiltonians on a Fock space is one of the most actively studied objects
in operator theory, in many problems in mathematical physics and other related fields.
In the present paper we consider a family of $2 \times 2$ operator matrices
${\mathcal A}_\mu(k),$ $k \in {\Bbb T}^3:=(-\pi, \pi]^3,$ $\mu>0,$
(so called generalized Friedrichs models)
associated with the Hamiltonian of a system consisting
of at most two particles on a three-dimensional lattice ${\Bbb Z}^3,$ interacting via creation and annihilation
operators. They are acting in the direct sum of zero-particle and one-particle subspaces of a Fock space.
The main goal of the paper is to give a thorough mathematical treatment of the spectral properties
of this family in dimension three with emphasis on threshold energy expansions
for the associated Fredholm determinant. More exactly,
we prove that there is a value $\mu_0$ of the parameter $\mu$ such that only for $\mu=\mu_0$ the
operator ${\mathcal A}_\mu(\overline{0})$ has a virtual level at the
point $z=0=\min\sigma_{\rm ess}({\mathcal A}_\mu(\overline{0}))$ and the operator ${\mathcal A}_\mu(\overline{\pi})$ has a virtual level at the
point $z=18=\max\sigma_{\rm ess}({\mathcal A}_\mu(\overline{\pi}))$,
where $\overline{0}:=(0,0,0), \overline{\pi}:=(\pi,\pi,\pi) \in {\Bbb T}^3.$
The absence of the eigenvalues of ${\mathcal A}_\mu(k)$ for all values of $k$ under the assumption
that $\mu=\mu_0$ is shown. The number, location and existence of the eigenvalues of ${\mathcal A}_\mu(k)$ are studied.
The threshold energy expansions for the Fredholm determinant associated
to ${\mathcal A}_\mu(k)$ are obtained. We point out that a part of
the results is typical for lattice models; in fact, they do not have analogues
in the continuous case (because its essential spectrum is half-line $[E; +\infty)$,
see for example \cite{MinSp96}).

We notice that threshold eigenvalue, virtual level (threshold energy resonance) and threshold energy expansion for the
associated Fredholm determinant of a generalized Friedrichs model have been
studied in \cite{ALR07, ALR07-1, TR11}. These results have been applied to the proof of the existence of Efimov's effect
and to obtain its discrete spectrum asymptotics. We note that above mentioned results are discussed
only for the bottom of the essential spectrum. The threshold eigenvalues and
virtual levels of a slightly simpler version of ${\mathcal
A}_\mu(k)$ were investigated in \cite{RD14}, and the
structure of the numerical range is studied using similar results.

The plan of this paper is as follows: Section 1 is an introduction
to the whole work. In Section 2, a family of $2 \times 2$ operator matrices are described as bounded self-adjoint
operators in the direct sum of two Hilbert spaces and its spectrum is described.
In Section 3, we discuss some results
concerning threshold analysis of a family of $2 \times 2$ operator matrices.

\section{\bf Family of $2\times2$ operator matrices and its spectrum}

It is well known that if ${\mathcal A}$ is a bounded linear operator
in a Hilbert space ${\mathcal H}$ and a decomposition ${\mathcal H}={\mathcal H}_0 \oplus {\mathcal H}_1$ is given,
then ${\mathcal A}$ always admits a block operator matrix
representation
\begin{equation*}
{\mathcal A}=\left( \begin{array}{cc}
A & B\\
C & D\\
\end{array}
\right)
\end{equation*}
with linear operators $A,$ $B,$ $C,$ and $D$ acting in or between the spaces
${\mathcal H}_0$ and ${\mathcal H}_1.$ It is easy to see that the operator ${\mathcal A}$ is a self-adjoint
if and only if $A=A^*,$ $D=D^*$ and $C=B^*.$

The present paper is devoted to the following case:
${\mathcal H}_0:={\Bbb C}$ is the field of complex
numbers (zero-particle subspace of a Fock space) and ${\mathcal H}_1:=L_2({\Bbb T}^3)$ is the Hilbert space of
square-integrable (complex-valued) functions defined on the three-dimensional torus ${\Bbb T}^3$
(one-particle subspace of a Fock space).

In the Hilbert space ${\mathcal H}:={\mathcal H}_0 \oplus {\mathcal H}_1$ we
consider the following family of $2\times2$ operator matrices
\begin{equation}\label{Frid model}
{\mathcal A}_\mu(k):=\left( \begin{array}{cc}
A_{00}(k) & \mu A_{01}\\
\mu A_{01}^* & A_{11}(k)\\
\end{array}
\right),
\end{equation}
where $A_{ii}(k): {\mathcal H}_i\to {\mathcal H}_i,$ $i=0,1,$ $k\in {\Bbb T}^3$ and
$A_{01}: {\mathcal H}_1 \to {\mathcal H}_0$ are defined by the rules
$$
A_{00}(k)f_0=w_0(k)f_0,\quad A_{01}f_1=
(f_1, 1), \quad
(A_{11}(k)f_1)(p)=w_1(k,p)f_1(p).
$$
Here $f_i \in {\mathcal H}_i,$ $i=0,1,$
$\mu>0$ is a coupling constant, the functions
$w_0(\cdot)$ and $w_1(\cdot, \cdot)$ have the form
$$
w_0(k):=\varepsilon(k)+\gamma, \quad
w_1(k,p):=\varepsilon(k)+\varepsilon(\frac{1}{2}(k+p))+\varepsilon(p)
$$
with $\gamma \in {\Bbb R}$ and the dispersion function $\varepsilon(\cdot)$
is defined by
\begin{equation}\label{epsilon}
\varepsilon(k):=\sum_{i=1}^3 (1-\cos \, k_i),\,k=(k_1, k_2, k_3) \in
{\Bbb T}^3,
\end{equation}
$A_{01}^*$ denotes the adjoint operator to $A_{01},$ that is,
$$
(A_{01}^*f_0)(p)=f_0, \quad f_0 \in {\mathcal H}_0.
$$
Under these assumptions the operator matrix ${\mathcal A}_\mu(k)$ is a bounded and
self-adjoint in ${\mathcal H}$.

We remark that the operators $A_{01}$ and $A_{01}^*$ are called
annihilation and creation operators \cite{Frid65}, respectively. In
physics, an annihilation operator is an operator that lowers the
number of particles in a given state by one, a creation operator is
an operator that increases the number of particles in a given state
by one, and it is the adjoint of the annihilation operator.

For convenience of the reader, we recall the notion of
the essential spectrum and the discrete spectrum of a bounded
self-adjoint operator.
Let ${\mathcal H}$ be a Hilbert space and ${\mathcal A}: {\mathcal H} \to {\mathcal H}$ be a bounded self-adjoint
operator. The set of all isolated eigenvalues of ${\mathcal A}$ with finite multiplicity is called
the discrete spectrum of ${\mathcal A}$ and denoted by $\sigma_{\rm disc}({\mathcal A})$. The set
$\sigma({\mathcal A}) \setminus \sigma_{\rm disc}({\mathcal A})$ is called an essential spectrum of
${\mathcal A}$ and denoted by $\sigma_{\rm ess}({\mathcal A})$.

The perturbation ${\mathcal A}_\mu(k)-{\mathcal A}_0(k)$ of the operator ${\mathcal A}_0(k)$ is a
self-adjoint operator of rank 2. Therefore, in accordance with the
invariance of the essential spectrum under the finite rank
perturbations, the essential spectrum $\sigma_{\rm ess}({\mathcal A}_\mu(k))$
of ${\mathcal A}_\mu(k)$ fills the following interval on the real axis
$$
\sigma_{\rm ess}({\mathcal A}_\mu(k))=[m(k), M(k)],
$$
where the numbers $m(k)$ and $M(k)$ are defined by
\begin{equation}\label{m(p) and M(p)}
m(k):=\min\limits_{p\in {\Bbb T}^3} w_1(k,p), \quad M(k):=
\max\limits_{p\in {\Bbb T}^3} w_1(k,p).
\end{equation}

For any $k\in {\Bbb T}^3$ we define an analytic function
$I(k\,; \cdot)$ in ${\Bbb C} \setminus
\sigma_{\rm ess}({\mathcal A}_\mu(k))$ by
$$
I(k\,; z):=\int_{{\Bbb T}^3} \frac{dt}{w_1(k,t)-z}.
$$
Then the Fredholm determinant associated to
the operator ${\mathcal A}_\mu(k)$ is defined by
\begin{equation*}
\Delta_\mu(k\,; z):=w_0(k)-z-\mu^2 I(k\,; z),\,\, z\in {{\Bbb C} \setminus
\sigma_{\rm ess}({\mathcal A}_\mu(k))}.
\end{equation*}

The following statement establishes connection
between the eigenvalues of the operator  ${\mathcal A}_\mu(k)$ and
zeros of the function $\Delta_\mu(k\,; \cdot)$, see \cite{ALR07, ALR07-1, TR11}.

\begin{lem}\label{Lemma 2.1.} For any $\mu>0$ and $k \in {\Bbb T}^3$ the operator ${\mathcal A}_\mu(k)$
has an eigenvalue $z_\mu(k) \in {\Bbb C} \setminus \sigma_{\rm ess}({\mathcal A}_\mu(k))$
if and only if $\Delta_\mu(k\,; z_\mu(k))=0$.
\end{lem}

From Lemma \ref{Lemma 2.1.} it follows that
$$
\sigma_{\rm disc}({\mathcal A}_\mu(k))=\{z\in{\Bbb C} \setminus
\sigma_{\rm ess}({\mathcal A}_\mu(k)):\,\Delta_\mu(k\,; z)=0\}.
$$

It is easy to show that the function $w_1(\cdot, \cdot)$ has an unique non-degenerate
minimum (resp. maximum) at the point $(\overline{0},\overline{0})\in ({\Bbb T}^3)^2$
(resp. $(\overline{\pi}, \overline{\pi})\in ({\Bbb T}^3)^2$) and
$$
\min\limits_{k,p\in {\Bbb T}^3}w_1(k,p)=w_1(\overline{0},\overline{0})=0,\quad
\max\limits_{k,p\in {\Bbb T}^3}
w_1(k,p)=w_2(\overline{\pi},\overline{\pi})=18,
$$
where $\overline{0}:=(0, 0, 0),\, \overline{\pi}:=(\pi, \pi, \pi) \in {\Bbb T}^3.$
Note that the function $w_0(\cdot)$ has also an unique
non-degenerate minimum (resp. maximum) at the point $\overline{0} \in
{\Bbb T}^3$ (resp. $\overline{\pi}\in {\Bbb T}^3$).
Simple calculations show that
\begin{align*}
& \sigma_{\rm ess}({\mathcal A}_\mu(\overline{0}))=[0; 9{\frac{3}{8}}];\\
& \sigma_{\rm ess}({\mathcal A}_\mu(\bar{\pi}))=[8{\frac{5}{8}}; 18];\\
& \sigma_{\rm disc}({\mathcal A}_\mu(\overline{0}))=\{z\in{\Bbb C}
\setminus [0; 9{\frac{3}{8}}]:\,\Delta_\mu(\overline{0}\,; z)=\gamma-z-\mu^2 I(\overline{0}\,; z)=0\};\\
&\sigma_{\rm disc}({\mathcal A}_\mu(\overline{\pi}))=\{z\in{\Bbb C}
\setminus [8{\frac{5}{8}}; 18]:\,\Delta_\mu(\overline{\pi}\,; z)=6+\gamma-z-\mu^2 I(\overline{\pi}\,; z)=0\}.
\end{align*}

Therefore,
$$
\min\limits_{k \in {\Bbb T}^3}\sigma_{\rm ess}({\mathcal A}_\mu(k))=0,\quad
\max\limits_{k \in {\Bbb T}^3}\sigma_{\rm ess}({\mathcal A}_\mu(k))=18.
$$

\section{\bf Eigenvalues, virtual levels and Fredholm determinant's expansions.}

In this Section we prove that there is a value $\mu_0$ of the parameter $\mu$ such that only for $\mu=\mu_0$ the
operator ${\mathcal A}_\mu(\overline{0})$ has a virtual level at the
point $z=0$ and the operator ${\mathcal A}_\mu(\overline{\pi})$ has a virtual level at the
point $z=18$.
The absence of the eigenvalues of ${\mathcal A}_\mu(k)$ for all values of $k$ under the assumption
that $\mu=\mu_0$ is shown. The number, location and existence of the eigenvalues of ${\mathcal A}_\mu(k)$ are studied.
The threshold energy expansions for the Fredholm determinant associated
to ${\mathcal A}_\mu(k)$ are obtained.

Using the extremal properties of the function $w_1(\cdot,\cdot)$, and the Lebesgue dominated convergence theorem
we obtain that there exists the positive limit
$$
\lim\limits_{z\to-0}\int_{{\Bbb T}^3} \frac{dt}{w_1(\overline{0},t)-z}=\int_{{\Bbb T}^3} \frac{dt}{w_1(\overline{0},t)}.
$$

For $\delta >0$ we set
$$
U_{\delta}(\overline{0}):=\{p \in {{\Bbb T}^3}: |p|<\delta \}.
$$

We show the finiteness of the integral
$$
\int_{{\Bbb T}^3} \frac{dt}{w_1(\overline{0},t)}.
$$
Since the function $w_1(\overline{0},\cdot)$ has an unique non-degenerate minimum at the point
$\overline{0} \in {\Bbb T}^3$,
there exist positive numbers $\delta, C_1, C_2$ such that
\begin{equation}\label{estimate for w1}
C_1|t|^2 \leq w_1(\overline{0},t)\leq C_2|t|^2,\quad  t\in {U_{\delta}(\overline{0})}.
\end{equation}

From the additivity of the integral it follows that
\begin{equation}\label{finite integral}
\int_{{\Bbb T}^3} \frac{dt}{w_1(\overline{0},t)}=\int_{{{\Bbb T}^3} \setminus {U_{\delta}(\overline{0})}}\frac{dt}{w_1(\overline{0},t)}+
\int_{U_{\delta}(\overline{0})} \frac{dt}{w_1(\overline{0},t)}.
\end{equation}

Since the integrand of the first summand on the r.h.s. of \eqref{finite integral} is continuous function on a compact set ${\Bbb T}^3 \setminus {U_{\delta}(\overline{0})}$, it is finite. Applying \eqref{estimate for w1} we deduce that
$$
\int_{U_{\delta}(\overline{0})} \frac{dt}{w_1(\overline{0},t)} \leq \frac{1}{C_1}\int_{U_{\delta}(\overline{0})} \frac{dt}{|t|^2}.
$$

Now, passing to the spherical coordinate system
\begin{align*}
& t_1=r \sin\psi \cos\varphi,\\
& t_2=r \sin\psi \sin\varphi,\\
& t_3=r \cos\psi,\quad 0 \leq r \leq \delta, \quad 0 \leq \varphi \leq 2\pi,\quad  0\leq \psi \leq \pi,
\end{align*}
we can assert that
$$
\int_{U_{\delta}(\overline{0})} \frac{dt}{|t|^2}=4\pi \delta<\infty.
$$
Set
\begin{align*}
& \mu_l^0(\gamma):=\sqrt{\gamma} \left(\int_{{\Bbb T}^3} \frac{dt}{w_1(\overline{0},t)} \right)^{-1/2}\,\, \mbox{for}\,\, \gamma>0;\\
& \mu_r^0(\gamma):=\sqrt{12-\gamma} \left(\int_{{\Bbb T}^3} \frac{dt}{w_1(\overline{0},t)} \right)^{-1/2}\,\, \mbox{for}\,\, \gamma<12.
\end{align*}

\begin{rem}
By the definition of $\mu_l^0(\gamma)$ and $\mu_r^0(\gamma)$ one can conclude that\\
if $\gamma\in (0; 6),$ then $\mu_l^0(\gamma)<\mu_r^0(\gamma);$\\
if $\gamma=6$, then $\mu_l^0(\gamma)=\mu_r^0(\gamma);$\\
if $\gamma\in (6; 12),$ then $\mu_l^0(\gamma)>\mu_r^0(\gamma).$
\end{rem}

Denote by $C({\Bbb T}^3)$ and $L_1({\Bbb T}^3)$ the Banach spaces of continuous and integrable functions on ${\Bbb T}^3$, respectively.

\begin{defn}\label{Definition 3.1.}
Let $\gamma \neq 0.$ The operator ${\mathcal A}_\mu(\overline{0})$ is said to have a virtual level at $z=0$
$($or zero-energy resonance$)$, if the number $1$
is an eigenvalue of the integral operator
$$
(G_{\mu}\psi)(q)=\frac{\mu^2}{\gamma} \int_{{\Bbb T}^3} \frac{\psi(t)dt}{\varepsilon(t/2)+\varepsilon(t)},\quad  \psi\in C({\Bbb T}^3)
$$
and the associated eigenfunction $\psi(\cdot)$ $($up to constant factor$)$ satisfies the condition $\psi(\overline{0}) \neq 0.$
\end{defn}

\begin{defn}\label{Definition 3.2.}
Let $\gamma \neq 12.$ The operator ${\mathcal A}_\mu(\overline{\pi})$ is said to have a virtual level at $z=18$, if the number $1$ is an eigenvalue of the integral operator
$$
(G'_{\mu}\varphi)(q)=\frac{\mu^2}{\gamma-12} \int_{{\Bbb T}^3} \frac{\varphi(t)dt}
{\varepsilon((\overline{\pi}+t)/2)+\varepsilon(t)-12},\quad  \varphi\in C({\Bbb T}^3)
$$
and the associated eigenfunction $\varphi(\cdot)$ $($up to constant factor$)$
satisfies the condition $\varphi(\overline{\pi})\neq 0.$
\end{defn}

\begin{rem}
The number $1$ is an eigenvalue of $G_{\mu}$ $($resp. $G'_{\mu})$ if and only if $\mu=\mu_l^0(\gamma)$ $($resp. $\mu=\mu_r^0(\gamma))$.
Consequently, the operator ${\mathcal A}_\mu(\overline{0})$ $($resp. ${\mathcal A}_\mu(\overline{\pi}))$ has a virtual level at $z=0$ $($resp. $z=18)$ if and only if $\mu=\mu_l^0(\gamma)$ $($resp. $\mu=\mu_r^0(\gamma))$.
\end{rem}

We notice that in the Definition \ref{Definition 3.1.}, the requirement of the presence of an eigenvalue $1$ of $G_{\mu}$ corresponds to the existence
of a solution of the equation ${\mathcal A}_\mu(\overline{0})f=0$
and the condition $\psi(\overline{0}) \neq 0$ implies that the solution $f=(f_0, f_1)$ of this equation does not belong to ${\mathcal H}$.
More exactly, if the operator ${\mathcal A}_\mu(\overline{0})$ has a virtual level at $z=0$, then the
vector-function $f=(f_0, f_1)$, where
$$
f_0={\rm const} \neq 0,\quad f_1(q)=-\frac {\mu f_0}{\varepsilon(q/2)+\varepsilon(q)}
$$
satisfies the equation ${\mathcal A}_\mu(\overline{0})f=0$ and $f_1\in L_1({\Bbb T}^3)\setminus L_2({\Bbb T}^3)$.

Indeed. The finiteness of the integral
$$
\int_{{\Bbb T}^3}|f_1(t)|dt={\mu}|f_0|\int_{{\Bbb T}^3}\frac{dt}{w_1(\overline{0},t)}
$$
is shown above. It follows that $f_1\in L_1({\Bbb T}^3)$. Using two-sided estimates \eqref{estimate for w1} we conclude that
$$
\int_{{\Bbb T}^3}|f_1(t)|^2 dt \geq \frac{{\mu}^2|f_0|^2}{C_2^2}\int_{U_\delta (\overline{0})}\frac{dt}{|t|^4}=\infty,
$$
and hence $f_1 \not\in L_2({\Bbb T}^3)$. It yields $f_1\in L_1({\Bbb T}^3)\setminus L_2({\Bbb T}^3)$.

Analogously, if the operator ${\mathcal A}_\mu(\overline{\pi})$ has a virtual level at $z=18$, then the
vector-function $f=(f_0, f_1)$, where
$$
f_0={\rm const} \neq 0,\quad f_1(q)=-\frac {\mu f_0}{\varepsilon((\overline{\pi}+q)/2)+\varepsilon(q)-12}
$$
obeys the equation ${\mathcal A}_\mu(\overline{\pi})f=18f$ and $f_1\in L_1({\Bbb T}^3)\setminus L_2({\Bbb T}^3)$.

\begin{thm}\label{Theorem 3.1.}
$(i)$ If $\gamma \leq 0$, then for any $\mu>0$ the operator ${\mathcal A}_\mu(\overline{0})$ has an unique negative eigenvalue.\\
$(ii)$ Let $\gamma>0$.\\
$(ii_1)$ For any $\mu \in (0; \mu_l^0(\gamma))$ the operator ${\mathcal A}_\mu(\overline{0})$ has no negative eigenvalues;\\
$(ii_2)$ If $\mu=\mu_l^0(\gamma)$, then the operator ${\mathcal A}_\mu(\overline{0})$ has a virtual level at the point $z=0$;\\
$(ii_3)$ For any $\mu>\mu_l^0(\gamma)$ the operator ${\mathcal A}_\mu(\overline{0})$ has an unique negative eigenvalue.
\end{thm}

\begin{proof}
$(i)$ Let $\gamma\leq 0$. Then for any $\mu>0$ the inequality
$$
\Delta_{\mu}(\overline{0}, 0)=\gamma-\mu^2{\int_{{\Bbb T}^3} \frac{dt}{w_1(\overline{0},t)}} \leq -\mu^2 \int_{{\Bbb T}^3} \frac{dt}{w_1(\overline{0}, t)}<0
$$
holds, that is, $\Delta_\mu(\overline{0}, 0)<0.$

It is easy to see that
$$
\lim\limits_{z\to-\infty} \Delta_\mu(\overline{0}, z)=\lim\limits_{z\to-\infty}
\left(\gamma-z-
\mu^2 \int_{{\Bbb T}^3} \frac{dt}{w_1(\overline{0},t)-z} \right)=+\infty.
$$

Since the function
$\Delta_{\mu}(\overline{0}, \cdot)$ is continuous and monotonically decreasing function on $(-\infty; 0),$
there exists a point $z_0(\mu) \in(-\infty; 0)$, such that,
$\Delta_{\mu}(\overline{0}, z_0(\mu))=0$. By Lemma \ref{Lemma 2.1.} the number $z_0(\mu)$ is an eigenvalue of ${\mathcal A}_\mu(\overline{0}).$

Let $\gamma>0$. $(ii_1)$. We assume that $\mu \in (0; \mu_l^0(\gamma))$.
For any $z\in(-\infty; 0)$ we have $\Delta_{\mu}(\overline{0}, z)>\Delta_{\mu}(\overline{0}, 0)$ and
$$
\Delta_{\mu}(\overline{0}, 0)=\gamma-\mu^2{\int_{{\Bbb T}^3} \frac{dt}{w_1(\overline{0},t)}}>\gamma-({\mu_l^0(\gamma)})^2{\int_{{\Bbb T}^3} \frac{dt}{w_1(\overline{0},t)}}=0.
$$
Therefore, $\Delta_{\mu}(\overline{0}, z)>0$ for any $z\in(-\infty; 0)$, that is, by Lemma \ref{Lemma 2.1.} the operator ${\mathcal A}_\mu(\overline{0})$
has no eigenvalues in $(-\infty; 0)$.

$(ii_2)$ Suppose that the operator ${\mathcal A}_\mu(\overline{0})$ has a virtual level
at $z=0.$ Then by Definition \ref{Definition 3.1.} the equation
$$
\psi(q)=\int_{{\Bbb T}^3} \frac{\psi(t)dt}{\varepsilon(t/2)+\varepsilon(t)},\quad  \psi\in C({\Bbb T}^3)
$$
has a nontrivial solution $\psi\in C({\Bbb T}^3)$, which satisfies the condition
$\psi(\overline{0}) \neq 0.$

This solution is equal to the function $\psi(q) \equiv 1$ (up to a constant factor)
and hence
$$
\Delta_\mu(\overline{0}, 0)=\gamma-\mu^2 \int_{{\Bbb T}^3} \frac{dt}{\varepsilon(t/2)+\varepsilon(t)}=0,
$$
that is, $\mu=\mu_l^0(\gamma).$

$(ii_3)$ Let now $\mu>\mu_l^0(\gamma)$. Then
$$
\Delta_{\mu}(\overline{0},0)=\gamma-\mu^2{\int_{{\Bbb T}^3}\frac{dt}{w_1(\overline{0},t)}<
\gamma-(\mu_l^0(\gamma))^2{\int_{{\Bbb T}^3}\frac{dt}{w_1(\overline{0},t)}}}=0,
$$
that is, $\Delta_{\mu}(\overline{0},0)<0$.
From
$$
\lim\limits_{z\to-\infty}{\Delta_{\mu}(\overline{0},z)}=+\infty
$$
we obtain that there exists
$z_0(\mu) \in(-\infty; 0)$ such that $\Delta_{\mu}(\overline{0}, z_0(\mu))=0.$ Again by Lemma \ref{Lemma 2.1.} the number $z_0(\mu)$ is an eigenvalue of
${\mathcal A}_\mu(\overline{0}).$
\end{proof}

The following Theorem may be proved in much the same way as Theorem \ref{Theorem 3.1.}

\begin{thm}\label{Theorem 3.2.}
$(i)$ If $\gamma \geq 12$, then for any $\mu>0$ the operator ${\mathcal A}_\mu(\overline{\pi})$ has no eigenvalues, bigger than $18.$\\
$(ii)$ Let $\gamma<12$.\\
$(ii_1)$ For any $\mu \in (0; \mu_r^0(\gamma))$ the operator ${\mathcal A}_\mu(\overline{\pi})$ has no eigenvalues, bigger than $18;$\\
$(ii_2)$ If $\mu=\mu_r^0(\gamma)$, then the operator ${\mathcal A}_\mu(\overline{\pi})$ has a virtual level at the point $z=18$;\\
$(ii_3)$ For any $\mu>\mu_r^0(\gamma)$ the operator ${\mathcal A}_\mu(\overline{\pi})$ has an unique eigenvalue in $(18;+\infty)$.
\end{thm}

Since $\mu_l^0(6)=\mu_r^0(6),$ setting $\mu_0:=\mu_l^0(6)$ from Theorems \ref{Theorem 3.1.} and \ref{Theorem 3.2.} we obtain the following

\begin{cor}
$(i)$ If $\gamma \in (0; 6)$, then for $\mu=\mu_l^0(\gamma)$
the operator ${\mathcal A}_\mu(\overline{0})$ has a virtual level at the point $z=0$
and the operator ${\mathcal A}_\mu(\overline{\pi})$ has no eigenvalues, bigger than $18;$\\
$(ii)$ If $\gamma=6$, then for $\mu=\mu_0$ the operators ${\mathcal A}_\mu(\overline{0})$ and
${\mathcal A}_\mu(\overline{\pi})$ have virtual levels at the points $z=0$ and $z=18$, respectively;\\
$(iii)$ If $\gamma \in (6; 12)$, then for $\mu=\mu_r^0(\gamma)$
the operator ${\mathcal A}_\mu(\overline{0})$ has an unique negative eigenvalue
and the operator ${\mathcal A}_\mu(\overline{\pi})$ has a virtual level at the point $18.$
\end{cor}

\begin{thm}\label{Theorem 3.3.}
For any $k\in {\Bbb T}^3$ the operator ${\mathcal A}_{\mu_0}(k)$ has no eigenvalues in
$(-\infty; 0) \bigcup(18; +\infty)$.
\end{thm}

\begin{proof}
Direct calculations shows that $\Delta_{\mu_0}(\overline{0}\,; 0)
= \Delta_{\mu_0}(\overline{\pi}\,; 18) = 0.$ The equality
$I(k\,; 0)=-I(k+\overline{\pi}\,; 18),$ $k \in {\Bbb T}^3$ implies that
$$
\Delta_{\mu}(\overline{0}\,; 0)=-\Delta_{\mu_0}(\overline{\pi}\,;
18)=6-\mu^2 I(\overline{0}\,; 0)=0
$$
and hence the last equality holds if and only if $\mu=\mu_0.$

First, we show that $I(k\,; 0)<I(\overline{0}\,; 0),$ $k \in {\Bbb T}^3
\setminus \{\overline{0}\}.$ Simple calculations shows that
$$
w_1(\overline{0},p)-\frac{w_1(k, p)+w_1(-k, p)}{2}=\sum_{i=1}^3 (\cos k_i-
1)+\sum_{i=1}^3 \cos \frac{p_i}{2} (\cos \frac{k_i}{2}-1).
$$
Then the equality
$$
I(k\,; 0)-I(\overline{0}\,; 0)=- \frac{1}{4}
\int_{{\Bbb T}^3}\frac{(w_1(k,t)-w_1(-k,t))^2}
{w_1(k,t) w_1(-k,t) w_1(\overline{0},t)}dt
$$
$$
+\int_{{\Bbb T}^3} \Bigl(w_1(\overline{0},t)-
\frac{w_1(k,t)+w_1(-k,t)}{2} \Bigr) \Bigl(\frac{1}{w_1(k,t) w_1(\overline{0},t)}
+\frac{1}{w_1(-k,t) w_1(\overline{0},t)} \Bigr)dt
$$
implies that $I(k\,; 0)<I(\overline{0}\,; 0)$ for all nonzero $k \in {\Bbb T}^3.$
The last inequality and the equality $I(k\,; 0)=-I(k+\overline{\pi}\,; 18),$
$k \in {\Bbb T}^3$ yields that $I(k\,; 18)>I(\overline{\pi}\,; 18),$ $k \neq \overline{\pi}.$ Taking into
account these inequalities together with the fact that the
function $w_0(\cdot)$ has an unique minimum at the point $\overline{0} \in
{\Bbb T}^3$ and maximum at the point $\overline{\pi}\in
{\Bbb T}^3$, we obtain that for any $\mu>0$ the function $\Delta_{\mu}(\cdot\,; 0)$
resp. $\Delta_{\mu}(\cdot\,; 18)$ has an unique minimum (resp. maximum)
at the point $\overline{0}\in {\Bbb T}^3$ (resp. $\overline{\pi}\in
{\Bbb T}^3$).

Since the function $\Delta_{\mu}(k\,; \cdot)$ is an increasing
function on $(-\infty; 0]\cup [18; +\infty)$ the relations
$$
\Delta_{\mu_0}(k\,; z)>\Delta_{\mu_0}(k\,; 0)\geq \Delta_{\mu_0}(\overline{0}\,; 0)=0,\quad z \in (-\infty, 0);
$$
$$
\Delta_{\mu_0}(k\,; z)<\Delta_{\mu_0}(k\,; 18)\leq
\Delta_{\mu_0}(\overline{\pi}\,; 18)=0,\quad z \in (18, \infty)
$$
hold. Hence, $\Delta_{\mu_0}(k\,; z)>0$ for $z\in (-\infty; 0)$ and $\Delta_{\mu_0}(k\,; z)<0$ for $z\in
(18; +\infty).$ This means that the function $\Delta_{\mu_0}(k\,; \cdot)$
has no zeros in $z\in (-\infty; 0)\cup (18; +\infty).$
Therefore, by Lemma \ref{Lemma 2.1.} the operator
${\mathcal A}_{\mu_0}(k),$ $k\in {\Bbb T}^3,$ has no eigenvalues in
$(-\infty; 0)\cup (18; +\infty).$
\end{proof}

Now we formulate and prove a result (threshold energy expansions for the
Fredholm determinant) of the paper, which is an important in the spectral analysis
for a $2 \times 2$ operator matrix acting in the direct sum of one-particle and two-particle subspaces of a Fock space
\cite{ALR07, ALR07-1, TR11}.

\begin{thm}\label{Theorem 3.4}
The following decompositions hold:
\begin{equation}\label{decomp-1}
\Delta_{\mu_0}(k\,; z)=\frac{32 \pi^2 \mu_0^2}
{5\sqrt{5}}\,\sqrt{\frac{6}{5}|k|^2-2z}+ O(|k|^2)+O(|z|), \,
|k|\to 0, \, z \nearrow 0;
\end{equation}
\begin{equation}\label{decomp-2}
\Delta_{\mu_0}(k\,; z)=-\frac{32 \pi^2 \mu_0^2}
{5\sqrt{5}}\,\sqrt{\frac{6}{5}|k-\overline{\pi}|^2+2(18-z)}
\end{equation}
$$
+ O(|k- \overline{\pi}|^2)+O(|z-18|), \quad |k-\overline{\pi}|\to
0, \quad z \searrow 18.
$$
\end{thm}

\begin{proof}
We give a sketch of the proof. The definition of
$w_0(\cdot)$ implies that
\begin{equation}\label{decomp of w1}
w_0(k)=12+O(|k-\overline{\pi}|^2), \quad |k-\overline{\pi}| \to 0.
\end{equation}

Taking sufficiently small $\delta>0$ and using the additivity of
the integral we rewrite $\Delta_{\mu_0}(k\,; z)$ as
\begin{equation}\label{decomp-3}
\Delta_{\mu_0}(k\,; z)=w_0(k)-z-\mu_0^2
\int_{U_\delta(\overline{\pi})} \frac{dt}{w_1(k,t)-z}
\end{equation}
$$
-\mu_0^2 \int_{{\Bbb T}^3 \setminus
U_\delta(\overline{\pi}) } \frac{dt}{w_1(k,t)-z}.
$$
Since the function $w_1(\cdot, \cdot)$ has an unique
non-degenerate maximum at $(\overline{\pi}, \overline{\pi}) \in
({\Bbb T}^3)^2,$ we easily derive the following relations
$$
\int_{U_\delta(\overline{\pi}) } \frac{dt}{w_1(k,t)-z}=
\int_{U_\delta(\overline{\pi}) }
\frac{dt}{w_1(\overline{\pi},t)-z}+\frac{32 \pi^2 \mu_0^2}
{5\sqrt{5}}\,\sqrt{\frac{6}{5}|k-\overline{\pi}|^2+2(18-z)}
$$
$$
+ O(|k-\overline{\pi}|^2)+O(|z-18|),\quad |k-\overline{\pi}|\to
0, \quad z \searrow 18
$$
and
$$
\int_{{\Bbb T}^3 \setminus U_\delta(\overline{\pi}) }
\frac{dt}{w_1(k,t)-z}= \int_{{\Bbb T}^3 \setminus
U_\delta(\overline{\pi})} \frac{dt}{w_1(\overline{\pi},t)-z} +
O(|k-\overline{\pi}|^2)+O(|z-18|),
$$
$$
|k-\overline{\pi}|\to 0, \, z
\searrow 18.
$$
Substituting the last two expressions and (\ref{decomp of w1})
into (\ref{decomp-3}) and using the equality $\Delta_{\mu_0}(\overline{\pi}\,; 18)=0$ we obtain (\ref{decomp-2}).
In the same manner we can obtain the representation
(\ref{decomp-1}).
\end{proof}

\end{document}